\documentclass[a4paper,fleqn,11pt]{article}
\usepackage{amsmath,amssymb,enumerate,latexsym,fancybox}
\usepackage{amsthm,amscd,ascmac,color,newtxmath,newtxtext,graphicx,natbib,bm,authblk}

\theoremstyle{plain}
  \newtheorem{theorem}{Theorem.}[section]
  \newtheorem{lemma}{Lemma.}[section]
  \newtheorem{corollary}{Corollary.}[section]
  \newtheorem{proposition}{Proposition.}[section]
\theoremstyle{definition}

\theoremstyle{remark}

\usepackage{tikz}
\usetikzlibrary{positioning,automata}

\title{Classical and quantum walks on paths associated with exceptional Krawtchouk polynomials}
\author[1]{Hiroshi Miki\thanks{E-mail:hmiki@mc-jma.go.jp}}
\author[2]{Satoshi Tsujimoto}
\author[3]{Luc Vinet}
\affil[1]{ Meteorological College, Asahi-Cho, Kashiwa 277 0852,
Japan}
\affil[2]{Department of Applied Mathematics and Physics, Graduate School of Informatics,
Kyoto University, Sakyo-Ku, Kyoto 606 8501, Japan}
\affil[3]{Centre de recherches math\'{e}matiques, Universit\'{e} de Montr\'{e}al,
PO Box 6128, Centre-ville Station, Montr\'{e}al (Qu\'{e}bec), H3C 3J7,
Canada}
\begin{document}
\maketitle
 
\begin{abstract}
Classical and quantum walks on some finite paths are introduced.
It is shown that these walks have explicit solutions given in terms of exceptional Krawtchouk polynomials and their properties are explored.
In particular, fractional revival is shown to take place in the corresponding quantum walks.  
\end{abstract}

\section{Introduction}
The theory and the applications of exceptional orthogonal polynomials whose development was initiated in \cite{gomez2009extended,gomez2010extension,quesne2008exceptional, odake2009infinitely} keep being actively explored. We shall here focus on polynomials of the discrete type \cite{duran2014exceptional,duran2015higher}. Broadly speaking, exceptional discrete orthogonal polynomials are defined as forming complete systems of polynomials orthogonal with respect to a positive measure and as being eigenfunctions of a second order difference operator with the distinct feature in comparison to classical ensembles, that the exceptional families have gaps in their degrees. For recent overviews and additional references, the reader may consult \cite{gomez2020exceptional, duran2020exceptional}. This paper examines the bearing and the application of the exceptional Krawtchouk polynomials on two related fields: the birth and death processes and (continuous-time) quantum walks.

An important special category of continuous time Markov processes is that of the so-called birth and death processes (BDP) which track for instance the size of a univariate population \cite{feller}. The basic form of BDP has the time $t$ running from $0$ to $\infty$, the states labelled by the non-zero integers $\{n\}$ and transitions from the state $n$ occurring only between the adjacent states $n+1$ (birth) or $n-1$ (death) with probabilities that depend solely on the state $n$ as per the Markov property. Specifically, if the system is in the state $n$ at time $t$, the probability that, between $t$ and $t+\Delta t$, the transition $n\rightarrow{n+1}$ occurs is $\lambda_n \Delta t + o(\Delta t)$, the transition $n \rightarrow n-1$ takes place is $\mu_n \Delta t + o(\Delta t)$, more than one transitions are observed is $o(\Delta t)$ and no transition is seen is $1 - (\lambda_n + \mu_n) \Delta t + o(\Delta t)$.

The integral formula obtained by Karlin and McGregor \cite{karlin1957differential} for the transition probabilities $P_{nm}(t)$ between states $n$ and $m$ over the time period $t$ established the intimate link between BDP and orthogonal polynomials. (Expositions of this connection will be found in \cite{ismail1990birth,schoutens2012stochastic}.) This seminal result is central in the ongoing study of the role of orthogonal polynomials in stochastic processes and has prompted the examination of generalizations of birth and death processes that are underscored by orthogonal polynomials of various kinds. This is illustrated by the following sampling. Karlin and McGregor themselves \cite{karlin1975linear} and others \cite{khare2009rates,champagnat2012dirichlet,griffiths2016multivariate,fernandez2021quasi} used multivariate orthogonal polynomials to analyze multidimensional or composition BDP. For applications of matrix orthogonal polynomials to quasi-birth and death processes one may consult \cite{dette2008some} or \cite{grunbaum2008qbd,gruunbaum2008random} in the discrete time case and again for discrete time, random walks that multiple polynomials entail have been looked at recently in \cite{branquinho2021multiple}.

Pursuing in this vein, we shall address here the question of what generalized birth and death processes could exceptional polynomials allow to solve by focusing on the process associated to a family of univariate (level 2) exceptional Krawtchouk polynomials. As shall be seen, this will yield a model where the states are labelled by the integers belonging to the set $\{0, 1,\dots, N, N+3\}$ and the transitions take place between three neighbours.

We shall also pay attention to the quantum realm. There is a natural parallel discussed in \cite{grunbaum2013birth} between BDP and continuous time quantum walks on weighted paths when only neighboring vertices are dynamically linked; it arises by basically performing the Wick rotation ($t \rightarrow it$). The analog of the Karlin McGregor formula provides in this context the transition amplitude for the quantum walker to hop between the sites $n$ and $m$ during time $t$. Such processes arise as the one-excitation dynamics of $XY$ quantum spin chains with non-uniform nearest-neighbour couplings and can be simulated by the propagation of photons in evanescently coupled waveguide arrays \cite{perez2013coherent,chapman2016experimental,bosse2017coherent}. These quantum walks are hence particularly relevant in the design of (one-dimensional) devices that realize with a minimum of external controls the end-to-end perfect state transfer (PST), which amounts to situations when transition amplitudes of modulus one can be achieved at certain times \cite{kay2010perfect}. Of interest also for the purpose of generating entangled states, is the construction of systems enabling fractional revival (FR), the phenomenon where an initial state is reproduced periodically at a number of fixed locations \cite{robinett2004quantum}. Continuous-time are also useful in this endeavour.

Not surprisingly orthogonal polynomials play again a key role in these matters. First, the general construction of systems with nearest neighbour interactions that achieve PST and/or FR involves the solution of inverse spectral problems that relies on the theory of univariate orthogonal polynomials \cite{gladwell2005inverse, vinet2012construct, genest2016quantum}. Second, the determination of analytic models with the advantages that these have, calls upon special polynomials that are fully characterized and explicitly known. Hence, as in the classical case, a search for solvable quantum walks corresponding to various families of univariate polynomials has been undertaken. Those corresponding to the Krawtchouk and dual Hahn polynomials were the first ones found \cite{albanese2004mirror}. A model associated to a special case of the $q$-Racah polynomials was exhibited \cite{vinet2012construct} as well as one corresponding to the dual $-1$ Hahn polynomials \cite{vinet2012dual, coutinho2019perfect} - a representative family of the interesting sets obtained from $q\rightarrow -1$ limits of the Askey-Wilson polynomials and their relatives. (A standard reference for information on most of these polynomials is \cite{koekoek2010hypergeometric}.) As a matter of fact, the search for models with FR \cite{genest2016quantum, lemay2016analytic} has led to the discovery of orthogonal polynomials of the para-type (see also \cite{vinet2012krawtchouk}) which arise through non-conventional truncations of infinite sets of polynomials belonging to the Askey scheme. A review of these various analytic systems is provided in \cite{bosse2017coherent}. Furthermore, as in the classical case, these enquiries suggested to look at the higher-dimensional quantum systems connected to multivariate polynomials \cite{miki2012quantum, miki2019quantum, miki2020perfect}.

We shall expand this body of work by examining also in this paper the transposition to the quantum domain of the generalized BDP associated to the exceptional Krawtchouk polynomials that we shall have introduced. Under some conditions, the corresponding quantum model will be shown to exhibit perfect return and some FR but PST will not be observed. 

The remainder of the paper is structured as follows. The essential properties of the exceptional $X_{\ell}$-Krawtchouk polynomials $\hat{K}_n ^{(\ell )}(x;p)$ will be recorded in Section 2. On the basis of the study \cite{miki2015new}, it will be stressed that these polynomials obey a $2\ell + 3$-term recurrence relation; the $7$-term one (for $\ell =2$) will be given explicitly and will be central in the applications to be pursued. In Section 3, a generalized birth and death process is defined and shown to be solvable with the help of the $X_2$-Krawtchouk polynomials. Its stationary distribution is provided. Section 4 examines the quantum walk generated by the Hamiltonian taken to be the $7$-diagonal matrix corresponding to the recurrence relation of the orthonormalized $X_2$-Krawtchouk polynomials. The transition amplitudes are expressed in terms of those polynomials and their properties are studied. Perfect return and fractional revival are shown to occur when the parameter $p$ is equal to $\frac{1}{2}$. Examples of dynamical evolution are also depicted graphically. Section 5 comprises concluding remarks.

\section{Ordinary and Exceptional Krawtchouk polynomials}
In this section, we first review the ordinary Krawtchouk polynomials. We then introduce the exceptional Krawtchouk polynomials and examine their properties. 
\subsection{Krawtchouk polynomials}
Let $0<p<1$ and $N$ be a positive integer. The Krawtchouk polynomials are usually defined in terms of the hypergeometric series \cite{koekoek2010hypergeometric}:
\begin{equation}
K_n^N(x;p)={}_{2}F_{1}\left(
  \begin{matrix}
   -n,-x\\
   -N
  \end{matrix}
  ;\frac{1}{p}\right)
\end{equation}
for $n,x=0,1,\dots ,N$. 
They are orthogonal with respect to the binomial distribution function:
\begin{equation}
\mathcal{L}[K_m^N(x;p)K_n^N(x;p)w_x]=\sum_{x=0}^N K_m^N(x;p)K_n^N(x;p)w_x=h_n\delta_{mn},
\end{equation}
where $\delta_{mn}$ is the Kronecker delta and
\begin{align}
\begin{split}
h_n&=\frac{(-1)^nn!}{(-N)_n}\left( \frac{q}{p}\right)^n,\quad (q=1-p),\\
w_x&=\binom{N}{x}p^xq^{N-x}=\frac{N!}{x!(N-x)!}p^xq^{N-x}.
\end{split}
\end{align}
One can easily see that the map $(p,x)\mapsto (q,N-x)$ does not change the weight function $w_x$. We thus find from the uniqueness of OPs that 
\begin{equation}\label{kraw:sym}
K_n^N(x;q)=\left( -\frac{p}{q}\right)^{n}K_n^N(N-x;p),\quad n=0,1,\ldots ,N.
\end{equation}
It should be remarked here that the relation \eqref{kraw:sym} holds for all $x,N \in \mathbb{Z}$.
Krawtchouk polynomials belong to the class of discrete classical OPs; they satisfy the following Sturm-Liouville difference equation:
\begin{equation}\label{kraw:diff}
-nK_n^N(x;p)=p(N-x)K_n^N(x+1;p)-\{ p(N-x)+qx\}K_n^N(x;p)+qxK_n^N(x-1;p)
\end{equation}
and from the exchange $n\leftrightarrow x$, one gets their recurrence relations:
\begin{equation}\label{kraw:rec}
-xK_n^N(x;p)=p(N-n)K_{n+1}^N(x;p)-\{ p(N-n)+qn\}K_n^N(x;p)+qnK_{n-1}^N(x;p).
\end{equation}
The forward shift operator relation for Krawtchouk polynomials will prove useful:
\begin{equation}\label{kraw:forward}
K_n^N(x+1;p)-K_n^N(x;p)=-\frac{n}{Np}K_n^{N-1}(x-1;p).
\end{equation}

\subsection{Exceptional Krawtchouk polynomials}
Let $\ell $ be a positive integer and introduce the index set
\begin{equation}
\Lambda_\ell  = \{ 0,1,\ldots , N,N+\ell +1\}.
\end{equation}
Furthemore, we consider the grid points
\begin{equation}
X_N=\{ -1,0,1,\ldots ,N\},
\end{equation}
which implies that we consider the space 
\begin{equation}
\{ \pi (x) \mod (x+1)x(x-1)\cdots (x-N)~|~\pi (x)\in \mathbb{R}[x]\}.
\end{equation}
For $x\in X_N$ and $n\in \Lambda_{\ell } $, the exceptional Krawtchouk polynomials of degree $\ell $ ($X_{\ell }$-Krawtchouk polynomials) can be defined via the Darboux transformation \cite{mtv}:
\begin{equation}\label{ekraw:def}
 \hat{K}^{(\ell )}_n(x;p)=
 \begin{cases}
{\cal F}_{N}[K_n^N(x;p)]   &  (n \ne N+\ell +1)\\
\noalign{\vskip 2mm}
\lim_{M\to N}{\cal F}_{M}[K_n^M(x;p)]   &  (n = N+\ell +1)\\
 \end{cases},
\end{equation}
where
\begin{equation}
  {\cal F}_{N} = (N-x) f_{\ell}^p(x)\,T +(1+x) f_{\ell}^p(x+1)\, I,
\end{equation}
with $f^p_{\ell }(x)=K^{-N-2}_{\ell }(x-N-1;p)$, $If(x)=f(x)$ and $Tf(x)=f(x+1)$.\\
The $X_{\ell}$-Krawtchouk polynomials satisfy the following Sturm-Liouville difference equation:
\begin{equation}
 {\cal F}_{N}  \circ {\cal B}_{N} [\hat{K}_n^{N,\ell}(x;p)] = (N+\ell+1 -n ) \hat{K}_n^{N,\ell}(x;p)
\end{equation} 
where
\begin{align}
 {\cal B}_{N} = \dfrac{1}{f_{\ell}^p(x)} \left(p I + q\, T^{-1}\right),
\end{align}
with $T^{-1}f(x)=f(x-1)$.
Their orthogonality relation is given by 
\begin{equation}\label{ekraw:orthogonality}
\hat{\mathcal{L}}[\hat{K}^{(\ell )}_m(x;p)\hat{K}^{(\ell )}_n(x;p)]=\sum_{x=-1}^N\hat{K}^{(\ell )}_m(x;p)\hat{K}^{(\ell )}_n(x;p)\hat{w}_x=\hat{h}_n\delta_{mn}
\end{equation}
with
\begin{align}
\begin{split}
\hat{h}_n&=\frac{(-1)^nn!}{(-N)_n}\left( \frac{q}{p}\right)^n(N+1)(N+l-n+1),\\
\hat{w}_x&= \binom{N+1}{x+1}\frac{p^{x+1}q^{N-x}}{f^p_{\ell }(x)f^p_{\ell }(x+1)}.
\end{split}
\end{align}
It should be mentioned here that the weight function $\hat{w}_x$ is positive when $l$ is even.
Like the Krawtchouk polynomials, the $X_l$-Krawtchouk polynomials are also symmetric with respect to $x$ and $p$.
\begin{proposition}
\begin{equation}\label{ekraw:sym}
\hat{K}^{(\ell )}_n(x;q)=\left( -\frac{p}{q}\right)^{n+l}\hat{K}^{(\ell )}_n(N-x-1;p)
\end{equation}
\end{proposition}
\begin{proof}
From relation \eqref{kraw:sym}, one can check that
\begin{align}
\begin{split}
f^q_{\ell }(x)&=K^{-N-2}_l(x-N-1;q)\\
&=\left( -\frac{p}{q}\right)^{\ell } K^{-N-2}_{\ell }(-N-2-(x-N-1);p)\\
&=\left( -\frac{p}{q}\right)^{\ell } f_{\ell }^p(N-x).
\end{split}
\end{align}
Therefore we have
\begin{align}
\begin{split}
&\hat{K}^{(\ell )}_n(x;q)\\
&=(N-x)f^q(x)K_n^N(x+1;q)+(1+x)f^q(x+1)K_n^N(x;q)\\
&=\left( -\frac{p}{q}\right)^{n+\ell }\left\{ (N-x)f^p_{\ell }(N-x)K_n^N(N-x-1;p)+ (1+x)f^p_{\ell }(N-x-1)K_n^N(N-x;p)\right\}\\
&=\left( -\frac{p}{q}\right)^{n+\ell }\hat{K}^{(\ell )}_n(N-x-1;p).
\end{split}
\end{align}
This completes the proof.
\end{proof}


It is known that exceptional orthogonal polynomials verify recurrence relations with more terms than those of the standard OPs \cite{miki2015new}.
In the case of the $X_{\ell }$-Krawtchouk polynomials, we have $2\ell +3$-term recurrence relations.
\begin{proposition}
For the $X_{\ell }$-Krawtchouk polynomials $K^{(l)}_n(x;p)$, the following relation holds:
\begin{equation}\label{ekraw:rec}
K^{-N-1}_{\ell +1}(x-N)\hat{K}^{(\ell )}_n(x)=\sum_{k=-\ell -1}^{\ell +1}\beta_{n,k}\hat{K}^{(\ell )}_{n+k}(x) \quad (\exists \beta_{n,k}\in \mathbb{R}).
\end{equation}
for $x\in X_N$.
\end{proposition}

We give the explicit form of the recurrence relation \eqref{ekraw:rec} when $\ell =2$:
\begin{align}\label{ekraw:rec2}
\begin{split}
 \lambda_x\hat{K}_n^{(2)}(x;p)&= \alpha_n\hat{K}_{n+3}^{(2)}(x;p)+\beta_n \hat{K}_{n+2}^{(2)}(x;p)+\gamma_n\hat{K}_{n+1}^{(2)}(x;p)\\
 &-(\alpha_n+\beta_n+\gamma_n+\delta_n+\epsilon_n+\zeta_n)\hat{K}_n^{(2)}(x;p)\\
 &+\delta_n \hat{K}_{n-1}^{(2)}(x;p)+\epsilon_n\hat{K}_{n-2}^{(2)}(x;p)+\zeta_n \hat{K}_{n-3}^{(2)}(x;p)
 \end{split}
 \end{align}
with
\begin{align}\label{ekraw:rec2-par}
\begin{split}
 \lambda_x&=-(K^{-N-1}_3(x-N;p)-K^{-N-1}_3(-1-N;p))\\
 \alpha_n&=\frac{(N-n+3)(N-n-2)_2}{(N+1)_3},\\
 \beta_n&=\frac{3(N-n+3)(N-n-1)_2(q-p)}{p(N+1)_3},\\
 \gamma_n&=\frac{3(N-n+3)(N-n)
 \left\{ N-n+1-(4N-5n+2)pq\right\}}{p^2(N+1)_3},\\
 \delta_n&=\frac{3n(N-n+3)q
 \left\{ N-n+2-(4N-5n+7)pq\right\}}{p^3(N+1)_3},\\
 \epsilon_n&=\frac{3(N-n+3)(n-1)_2(q-p)q^2}{p^3(N+1)_3},\\
 \zeta_n&=\frac{(n-2)_3q^3}{p^3(N+1)_3}.  
\end{split}
\end{align}

\section{Classical random walks associated with exceptional Krawtchouk polynomials}
In this section, we introduce some classical random walks which can be regarded as generalized birth and death processes. 
We then show that these are exactly solvable with a parametrization derived from the exceptional Krawtchouk polynomials.\par 
We consider stationary Markov processes whose state space is $\Lambda_2=\{ 0,1,\ldots ,N,N+3\}$.
With $i,j \in \Lambda _2$, we assume that the probabilities $P_{ij}(\Delta_t)$ for the system to evolve from state $i$ to $j$ in infinitesimal time $\Delta_t$ are given by:
\begin{align}\label{gbdp:def}
p_{ij}(\Delta_t)=\begin{cases}
\alpha_i\Delta_t+o(\Delta_t), &j=i+3\\
\beta_i\Delta_t+o(\Delta_t), &j=i+2\\
\gamma_i\Delta_t+o(\Delta_t), &j=i+1\\
1-(\alpha_i+\beta_i+\gamma_i +\delta _i+\epsilon_i+\zeta_i )\Delta_t+o(\Delta_t), &j=i\\
\delta_i\Delta_t+o(\Delta_t), &j=i-1\\
\epsilon_i\Delta_t+o(\Delta_t), &j=i-2\\
\zeta_i\Delta_t+o(\Delta_t), &j=i-3\\
o(\Delta_t), & |i-j|>3
\end{cases},
\end{align}
where $\alpha_i,\beta_i,\gamma_i,\delta_i,\epsilon_i,\zeta_i \ge 0$.
Unlike the ordinary birth and death processes, here as shown in Fig. \ref{gbdp}, the transition occurs among three-nearest neighbors although state $N+3$ can only be exchanged with state $N$.
\begin{figure}[htbp]
\begin{center}
\begin{tikzpicture}
\tikzset{ccc/.style={circle, fill=blue, text=white, text width=0.8cm, text centered, rounded corners, minimum height=0.5cm}};
    \node[ccc] (n-3) {$n-3$};
    \node[ccc,right=0.5cm of n-3] (n-2) {$n-2$};
    \node[ccc,right=0.5cm of n-2] (n-1) {$n-1$};
    \node[ccc,right=0.5cm of n-1] (n) {$n$};
    \node[ccc,right=0.5cm of n] (n+1) {$n+1$};
    \node[ccc,right=0.5cm of n+1] (n+2) {$n+2$};
    \node[ccc,right=0.5cm of n+2] (n+3) {$n+3$};
    \node[right=0.2cm of n+3] (l) {$\cdots $};
    \node[ccc,right=1cm of n+3] (N) {$N$};
    \node[ccc,right=0.5cm of N] (N+3) {$N+3$};
	\path [->,ultra thick] (n) edge [bend left]  node [above] {$\gamma_n$} (n+1)
	(n) edge [bend left]  node [above] {$\beta_n$} (n+2)
	(n) edge [bend left]  node [above] {$\alpha_n$} (n+3)
	(N) edge [bend left]  node [above] {$\alpha_N$} (N+3);
	\path [->,ultra thick] (n) edge [bend left]  node [below] {$\delta_n$} (n-1)
	(n) edge [bend left] node [below] {$\epsilon_{n}$} (n-2)
	(n) edge [bend left]  node [below] {$\zeta_n$} (n-3)
	(N+3) edge [bend left] node [below] {$\zeta_{N+3}$} (N);
	
\end{tikzpicture}
\end{center}
\caption{The conceptual image of the generalized birth and death processes.}
\label{gbdp}
\end{figure}
By hypothesis the conditional probability of the process up to time $t$ depends only on the state of the process at time $t$, which means that 
\begin{equation}
P_{ij}(s+t)=\sum_{k\in \Lambda_2} P_{ik}(s)P_{kj}(t).
\end{equation}
This results in 
\begin{equation}
P(s+t)=P(s)P(t),\quad P(t)=(P_{ij}(t))_{i,j\in \Lambda_2}.
\end{equation}
Furthermore, it is apparent from the definition that $P(0)=I$ is the identity matrix.
From this fact and \eqref{gbdp:def}, one has 
\begin{align}
\begin{split}
&P_{ij}(t+\Delta_t)=\sum_{k\in \Lambda_2}P_{ik}(t)P_{kj}(\Delta_t)\\
&=P_{i,j-3}(t)P_{j-3,j}(\Delta_t)+ P_{i,j-2}(t)P_{j-2,j}(\Delta_t)+P_{i,j-1}(t)P_{j-1,j}(\Delta_t)\\
&+P_{i,j+1}(t)P_{j+1,j}(\Delta_t)+P_{i,j+2}(t)P_{j+2,j}(\Delta_t)+P_{i,j+3}(t)P_{j+3,j}(\Delta_t)\\
&+P_{i,j}(t)P_{j,j}(\Delta_t)+o(\Delta_t),
\end{split}
\end{align}
from where we obtain the following (forward) Chapman-Kolmogorov differential equation:
\begin{align}\label{gbdp:diff}
\begin{split}
\frac{dP_{ij}(t)}{dt}&=\alpha_{j-3}P_{i,j-3}(t)+\beta_{j-2}P_{i,j-2}(t)+\gamma_{j-1}P_{i,j-1}(t)\\
&-(\alpha_j+\beta_j+\gamma_j+\delta_j+\epsilon_j+\zeta_j)P_{i,j}(t)\\
&+\delta_{j+1}P_{i,j+1}(t)+\epsilon_{j+2}P_{i,j+2}(t)+\zeta_{j+3}P_{i,j+3}(t).
\end{split}
\end{align}
The equation \eqref{gbdp:diff} can be cast in the following matrix form:
\begin{equation}
\frac{dP(t)}{dt}=P(t)A,
\end{equation}
where $A \in \mathbb{R}^{(N+2)\times (N+2)}$ is the seven-diagonal (except for last column and row) matrix:
\begin{equation}\label{gbdp:matrix}
A={\scriptsize
\begin{pmatrix}
-S_0 & \gamma_0 & \beta_0 & \alpha_0 & & &  & & & \\
\delta_1 & -S_1 & \gamma _1 & \beta_1 & \alpha _1 & & & & &  \\
\epsilon_2 & \delta_1 & -S_2 & \gamma_2 & \beta_2 & \alpha_2 & & & &  \\
\zeta_3 & \epsilon_3 & \delta_3 & -S_3 & \gamma_3 & \beta_3 & \alpha_3 & & & \\
        & \ddots & \ddots & \ddots & \ddots & \ddots & \ddots & \ddots & & \\
        &        & \zeta_{N-3} & \epsilon_{N-3} & \delta_{N-3} & -S_{N-3} & \delta_{N-3} & \beta_{N-3} & \alpha _{N-3}\\
        &        &     & \zeta_{N-2} & \epsilon_{N-2} & \delta_{N-2} & -S_{N-2} & \gamma_{N-2} & \beta_{N-2} &  \\
        &  & & & \zeta_{N-1} & \epsilon_{N-1} & \delta_{N-1} & -S_{N-1} & \gamma_{N-1} &  \\
        &  & & & & \zeta_N & \epsilon_N & \delta_N & -S_N & \alpha_N \\
        &  & & & &         & & & \zeta_{N+3} & -S_{N+3}        
\end{pmatrix}}
\end{equation} 
with $S_n=\alpha_n+\beta_n+\gamma_n+\delta_n+\epsilon_n+\zeta_n$.
Comparing \eqref{ekraw:rec2} and \eqref{gbdp:matrix}, this generalized birth and death process can be exactly solved under appropriate parametrization.
\begin{theorem}\label{thm:bdp}
Set the parameters $\alpha_n,\beta_n,\gamma_n,\delta_n,\epsilon_n,\zeta_n$ as \eqref{ekraw:rec2-par}, which are all non-negative for $0<p\le \frac{1}{2}$. Then the probability $P(t)=(P_{ij}(t))_{i,j\in \Lambda_2}$ is explicitly given in terms of the $X_2$-Krawtchouk polynomials $\{ \hat{K}_n^{(2)}(x;p)\}_{n\in \Lambda_2}$:
\begin{equation}\label{gbdp:sol}
P_{ij}(t)=\frac{1}{\hat{h}_j}\sum_{x=-1}^N \hat{w}_x  \hat{K}_i^{(2)}(x;p) \hat{K}_j^{(2)}(x;p)e^{\lambda_xt},
\end{equation}  
where
\begin{align}
\begin{split}
\lambda_x&=-(K_3^{-N-1}(x-N;p)-K_3^{-N-1}(-1-N;p)),\\
\hat{h}_j&=\frac{(-1)^jj!}{(-N)_j}\left( \frac{q}{p}\right)^j(N+1)(N+3-n),\\
\hat{w}_x&= \binom{N+1}{x+1}\frac{p^{x+1}q^{N-x}}{K_2^{-N-2}(x-N-1;p)K_2^{-N-2}(x-N;p)}.
\end{split}
\end{align}
\end{theorem}
\begin{proof}
The proof can be done following an approach similar to the one described in \cite{grunbaum2013birth}. 
The non-negativity of the parameters is easily checked since $q-p=1-2p\ge 0, N-n+1-(4N-5n+2)pq\ge \frac{n+2}{4},N-n+1-(4N-5n+2)pq,N-n+2-(4N-5n+7)pq\ge \frac{n+1}{4}$ for $0<p\le \frac{1}{2}$.
Let us introduce the function
\begin{equation}\label{ri}
r_{i}(x,t)=\sum_{j\in \Lambda_2}P_{ij}(t)\hat{K}_j^{(2)}(x;p),\quad i\in \Lambda_2.
\end{equation}
Differentiating both sides of \eqref{ri} with respect to $t$, we obtain from \eqref{ekraw:rec2}
\begin{align}
\begin{split}
\frac{\partial r_i(x,t)}{\partial t}=&\sum_{j\in \Lambda_2}\hat{K}_j^{(2)}(x;p)\frac{dP_{ij}(t)}{dt}\\
=&\sum_{j\in \Lambda_2}\hat{K}_j^{(2)}(x;p)\left\{ \alpha_{j-3}P_{i,j-3}(t)\right.\\
&+\beta_{j-2}P_{i,j-2}(t)+\gamma_{j-1}P_{i,j-1}(t)\\
&-(\alpha_j+\beta_j+\gamma_j+\delta_j+\epsilon_j+\zeta_j)P_{i,j}(t)\\
&\left.+\delta_{j+1}P_{i,j+1}(t)+\epsilon_{j+2}P_{i,j+2}(t)+\zeta_{j+3}P_{i,j+3}(t)\right\}\\
=&\sum_{j\in \Lambda_2}P_{ij}(t)\left\{ \alpha_j \hat{K}_{j+3}^{(2)}(x;p)+\beta_j \hat{K}_{j+2}^{(2)}(x;p)+\gamma_{j}\hat{K}_{j+1}^{(2)}(x;p)\right.\\
&-(\alpha_j+\beta_j+\gamma_j+\delta_j+\epsilon_j+\zeta_j)\hat{K}_{j}^{(2)}(x;p)\\
&\left.+\delta_{j}\hat{K}_{j-1}^{(2)}(x;p)+\epsilon_j\hat{K}_{j-2}^{(2)}(x;p)+\zeta_j\hat{K}_{j-3}^{(2)}(x;p)\right\}\\
=&\sum_{j\in \Lambda_2}P_{ij}(t)\lambda_x \hat{K}_{j}^{(2)}(x;p)\\
=&\lambda_x r_{i}(x,t).
\end{split}
\end{align} 
From the initial condition $r_i(x,0)=\hat{K}_{i}^{(2)}(x;p)$, one thus finds 
\begin{equation}
r_i(x,t)=e^{\lambda_x t}\hat{K}_{i}^{(2)}(x;p).
\end{equation}
Therefore we have
\begin{align}
\begin{split}
&\sum_{x=-1}^N \hat{w}_x \hat{K}_{i}^{(2)}(x;p)\hat{K}_{j}^{(2)}(x;p)e^{\lambda_xt}=\sum_{x=-1}^N r_i(x,t)\hat{K}_{j}^{(2)}(x;p) \hat{w}_x\\
&=\sum_{x=-1}^N \hat{w}_x  \left( \sum_{k\in \Lambda_2} P_{ik}(t)\hat{K}_{k}^{(2)}(x;p) \right) \hat{K}_{j}^{(2)}(x;p)\\
&=\sum_{k\in \Lambda_2} P_{ik}(t)\sum_{x=-1}^N \hat{w}_x  \hat{K}_{k}^{(2)}(x;p) \hat{K}_{j}^{(2)}(x;p)\\
&=\sum_{k\in \Lambda_2}P_{ik}(t)\hat{h}_{j}\delta_{kj}=P_{ij}(t)\hat{h}_j.
\end{split}
\end{align}
This completes the proof.
\end{proof}
Using Theorem \ref{thm:bdp}, it is not difficult to find the stationary distribution of the model (the limit $t\to \infty $) and to see that it does not depend on the initial state.
\begin{corollary}
For the probability \eqref{gbdp:sol}, the following limit holds:
\begin{equation}
\lim_{t\to \infty }P_{ij}(t)=r_j= \binom{N+3}{j} \frac{(N-j+1)_2p^{j-2}q^{N-j+3}}{(N+2)_2K_2^{-N-2}(-N-1;p)}.
\end{equation}
\end{corollary}

\section{Quantum walks associated with exceptional Krawtchouk polynomials}
In the previous section, we have introduced a generalized birth and death process that can be viewed as a classical random walk.
We now consider the quantum analogue of this walk described by the following Schr\"{o}dinger equation:
\begin{equation}\label{schrodinger}
i\hbar\frac{d}{dt}\left| \psi (t)\right> = M \left| \psi (t)\right>,
\end{equation}
where $M \in \mathbb{R}^{(N+2)\times (N+2)}$ is a weighted adjacency matrix associated to the $X_2$-Krawtchouk polynomials and $\left| \psi (t)\right> \in \mathbb{C}^{N+2}$ is a state vector.
The time variable is scaled so that $\hbar =1$.
Eq. \ref{schrodinger} implies the following evolution of the initial vector $\left| \psi (0)\right>$:
\begin{equation}
\left| \psi (t)\right> = \exp(-itM)\left| \psi (0)\right>.
\end{equation}
It is not difficult to see that $M$ must be symmetric so that $|\left < \psi (t)|\psi (t)\right>| =1$ for any time $t$.
The orthonormal $X_2$-Krawthcouk polynomials $\{ \bar{K}_n(x;p)\} _{n\in \Lambda_2 }$ thus come into play.
\begin{proposition}
For the orthonormal $X_2$-Krawtchouk polynomials
\begin{equation}
\bar{K}_n(x;p)=\frac{1}{\sqrt{h_n}}\hat{K}_n^{(2)}(x;p),
\end{equation}
the recurrence relation is found to be from \eqref{ekraw:rec}
\begin{align}\label{ekraw:rec-orthonormal}
\begin{split}
\bar{\lambda}_x\bar{K}(x;p)&=I_{n+3}\bar{K}_{n+3}(x;p)+J_{n+2}\bar{K}_{n+2}(x;p)+L_{n+1}\bar{K}_{n+1}(x;p)\\
 &+L_{n}\bar{K}_{n-1}(x;p)+J_{n}\bar{K}_{n-2}(x;p)+I_{n}\bar{K}_{n-3}(x;p)\\
 &-\bar{S}_{n}\bar{K}_n(x;p),
 \end{split}
\end{align}
where 
\begin{align}
\begin{split}
\bar{\lambda}_x&=-K^{-N-1}_3(x-N;p),\\
I_n&=\sqrt{\alpha_{n-3}\zeta_n}=\frac{\sqrt{(N-n+6)(N-n+1)_2(n-2)_3}}{(N+1)_3}\cdot \left( \frac{q}{p}\right)^{\frac{3}{2}},\\
J_n&=\sqrt{\beta_{n-2}\epsilon_n}=\frac{3\sqrt{(N-n+5)(N-n+1)_3(n-1)_2}}{(N+1)_3}\cdot \frac{q(q-p)}{p^2},\\
L_n&=\sqrt{\gamma_{n-1}\delta_n}\\
&=\frac{3\sqrt{n(N-n+1)(N-n+3)_2}}{p^2(N+1)_3}\cdot \left( \frac{q}{p}\right)^{\frac{1}{2}}\left\{  N-n+2-(4N-5n+7)pq\right\},\\
\bar{S}_n&=S_n+K^{-N-1}_3(-1-N;p).
\end{split}
\end{align}
\end{proposition}
Introduce the sequence of functions $\{ T_n(x;p)\}_{n=1}^{N+2}$ as follows
\begin{align}
\begin{split}
T_n(x;p)=
\begin{cases}
\sqrt{\frac{w_x}{\hat{h}_{n-1}}}\hat{K}_{n-1}^{(2)}(x;p)\quad &n=1,2,\ldots ,N+1\\
\sqrt{\frac{w_x}{\hat{h}_{N+3}}}\hat{K}_{N+3}^{(2)}(x;p)\quad &n=N+2
\end{cases}
\end{split}
\end{align}
It is straightforward to see from the orthogonality relation \eqref{ekraw:orthogonality} that
\begin{equation}\label{ekraw:orthonormality}
\sum_{x=-1}^N T_n(x;p)T_m(x;p)=\delta_{mn}.
\end{equation} 
Furthermore, from \eqref{ekraw:rec-orthonormal}, the vector
\begin{equation}
\bm{v}_x =\left( T_1(x;p),T_1(x;p),\ldots ,T_{N+2}(x;p)\right)^T
\end{equation}
is a normalized eigenvector of the following eigenvalue problems:
\begin{equation}\label{ekraw:evp}
M\bm{v}_x=\lambda_x \bm{v}_x,\quad x=-1,0,\ldots ,N
\end{equation}
with
\begin{equation}
M={\scriptsize
\begin{pmatrix}
-\bar{S}_0 & L_1 & J_2 & I_3 & & &  & & & \\
L_1 & -\bar{S}_1 & L_2 & J_3 & I_4 & & & & &  \\
J_2 & L_2 & -\bar{S}_2 & L_3 & J_4 & I_5 & & & &  \\
I_3 & J_3 & L_3 & -\bar{S}_3 & L_4 & J_5 & I_6 & & & \\
        & \ddots & \ddots & \ddots & \ddots & \ddots & \ddots & \ddots & & \\
        &        & I_{N-3} & J_{N-3} & L_{N-3} & -\bar{S}_{N-3} & L_{N-2} & J_{N-1} & I _{N}\\
        &        &     & I_{N-2} & J_{N-2} & L_{N-2} & -\bar{S}_{N-2} & L_{N-1} & J_{N} &  \\
        &  & & & I_{N-1} & J_{N-1} & L_{N-1} & -S_{N-1} & L_{N} &  \\
        &  & & & & I_N & J_N & L_N & -\bar{S}_N & I_{N+3} \\
        &  & & & &         & & & I_{N+3} & -\bar{S}_{N+3}        
\end{pmatrix}}.
\end{equation}
Using the $X_2$-Krawtchouk polynomials, we can calculate exactly the transition amplitude for a particle initially at site $i$ to be found at site $j$ at time $t$:
\begin{equation}\label{qrw:transiton-amplitude}
c_{ij}(t)= \left( e_j|\exp(-iMt)|e_i\right),\quad i,j=0,1,\cdots ,N+1,
\end{equation}
where $\left| e_j\right) =(0,0,\ldots ,0,1,0,\ldots ,0)^T$ is a unit vector with $1$ at $j$-th entry and zero elsewhere.
\begin{theorem}
The transition amplitude \eqref{qrw:transiton-amplitude} is explictly given by
\begin{align}
c_{ij}(t)=\sum_{x=-1}^N T_i(x;p)T_j(x;p)e^{-i\bar{\lambda}_x t}.
\end{align}
\end{theorem}
\begin{proof}
From the relation \eqref{ekraw:orthonormality}, it is not difficult to see that
\[
VV^T=E,\quad V=\left( \bm{v}_{-1},\bm{v}_1,\ldots, \bm{v}_{N},\bm{v}_{N}\right).
\]
Therefore, one finds from \eqref{ekraw:evp} that
\[
M=VDV^T,\quad D=\mathrm{diag}(\bar{\lambda}_{-1},\bar{\lambda}_0,\ldots ,\bar{\lambda}_{N+1})
\]
to obtain 
\begin{align}
\begin{split}
c_{ij}(t)&=\left( e_j| \exp(-iMt)|e_i\right)\\
&=\left( e_j| \exp(-iVDV^T t)|e_i\right)\\
&=\left( e_j| V\exp(-iDt)V^T |e_i\right)\\
&=\sum_{x=-1}^N T_i(x;p)T_j(x;p)e^{-i\bar{\lambda}_x t}.
\end{split}
\end{align}
This completes the proof.
\end{proof}
It is worth noting that the spectrum $\lambda_x$ is rational when $p$ is rational, which implies that there exists some time $t_0$ such that $(\bar{\lambda}_x-\bar{\lambda}_y)t_0$ is divisible by $2\pi $.
Therefore, as is discussed in \cite{godsil2012state}, perfect return will take place.
\begin{corollary}\label{cor:pr}
When $0< p\le \frac{1}{2}$ is rational, there exists some time $t_0$ such that 
\begin{align}
|c_{ij}(t_0)|=\delta_{ij}.
\end{align} 
\end{corollary}
\begin{proof}
Since $p$ is rational, there exists some time $t_0$ such that $(\bar{\lambda}_x-\bar{\lambda}_y)t_0$ is divisible by $2\pi $, which amounts to 
\begin{equation}
e^{-i\bar{\lambda_x} t_0}= e^{-i\bar{\lambda_y} t_0}=e^{-i\bar{\lambda}_{-1}t_0},\quad x,y=-1,0,\ldots ,N.
\end{equation}
Therefore we have from the orthogonality \eqref{ekraw:orthonormality} 
\begin{align}
\begin{split}
c_{ij}&=\sum_{x=-1}^N T_i(x;p)T_j(x;p)e^{-i\bar{\lambda}_x t_0}\\
&=e^{-i\bar{\lambda}_{-1}t_0}\sum_{x=-1}^N T_i(x;p)T_j(x;p)=e^{-i\bar{\lambda}_{-1}t_0} \delta_{ij}.
\end{split}
\end{align}
This completes the proof.
\end{proof}
For example, when $N=5$ and $p=\frac{1}{4}$, the matrix $M$ is 
\begin{equation}
M=\begin{pmatrix}
\frac{92}{7} & \frac{15\sqrt{210}}{28} & \frac{3\sqrt{30}}{7} & \frac{3\sqrt{3}}{14} & 0 & 0 & 0\\
\frac{15\sqrt{210}}{28} & \frac{25}{2} & \frac{87\sqrt{7}}{28} & \frac{9\sqrt{70}}{28} & \frac{3\sqrt{21}}{28} & 0 & 0\\
\frac{3\sqrt{30}}{7} & \frac{87\sqrt{7}}{28} & \frac{74}{7} & \frac{9\sqrt{10}}{4} & \frac{9\sqrt{3}}{7} & \frac{3\sqrt{15}}{28} & 0 \\
\frac{3\sqrt{3}}{14} & \frac{9\sqrt{70}}{28} & \frac{9\sqrt{10}}{4} & \frac{55}{7} & \frac{27\sqrt{30}}{28} & \frac{15\sqrt{6}}{28} & 0\\
0 & \frac{3\sqrt{21}}{28} & \frac{9\sqrt{3}}{7} & \frac{3\sqrt{15}}{28} & \frac{34}{7} & \frac{39\sqrt{5}}{28} & 0\\
0 & 0 & \frac{3\sqrt{15}}{28} & \frac{15\sqrt{6}}{28} & \frac{39\sqrt{5}}{28} & \frac{29}{14} & \frac{3\sqrt{42}}{28}\\
0 & 0 & 0 & 0 &0 &\frac{3\sqrt{42}}{28} & 0
\end{pmatrix}
\end{equation}
and its eigenvalues are given by
\begin{equation}
27,\quad \frac{103}{7},\quad 7,\quad \frac{19}{7},\quad \frac{5}{7},\quad -\frac{1}{7},\quad -1. 
\end{equation}
Therefore, perfect return will take place at time $t_0=7\pi $.
This is depicted in Fig. \ref{qw_p=1/4} 

\begin{figure}[htbp]
\begin{center}
\begin{tabular}{cl}
\includegraphics[width=8cm]{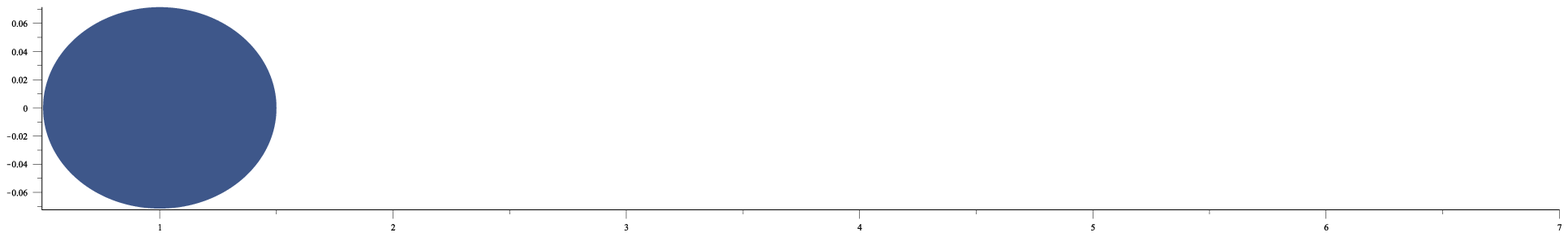}  & $(t=0)$\\
\includegraphics[width=8cm]{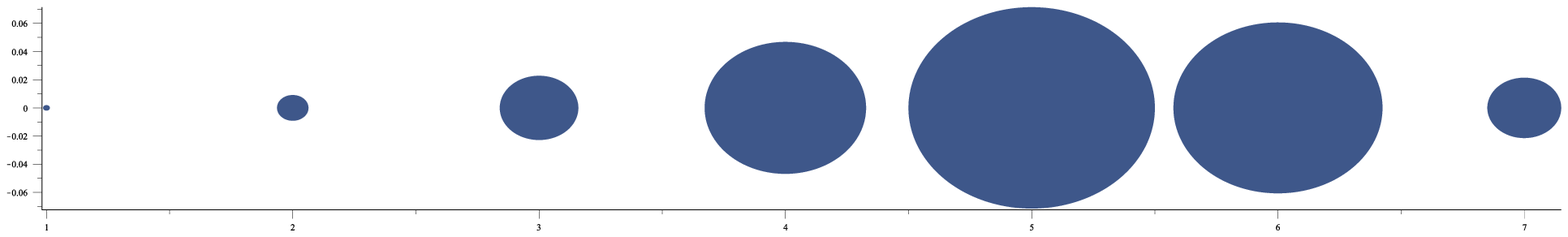} & $(t=\frac{7}{2}\pi)$\\
\includegraphics[width=8cm]{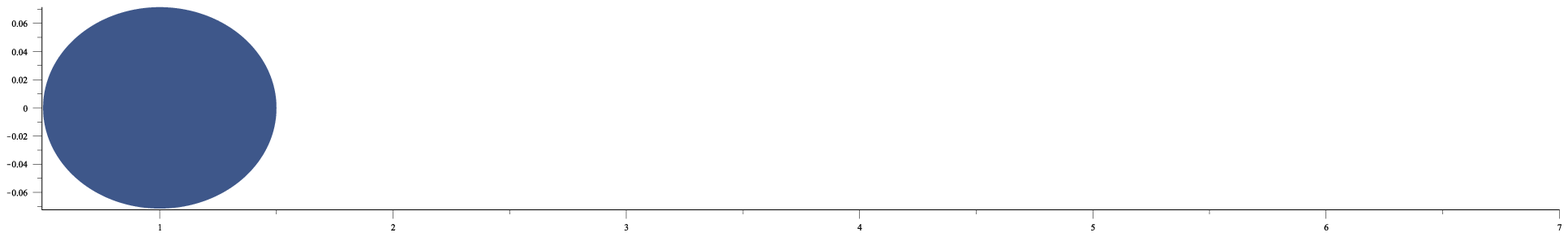} &  $(t=7\pi)$
\end{tabular}
\end{center}
  \caption{The plot of transition amplitude $\{ |c_{1i}(t)|\}_{i=0}^N$ with $N=5$ and $p=\frac{1}{4}$. The areas of the circles are proportional to $|c_{1j}(t)|$ at the given lattice point $j$. Perfect return from $1$ to $1$ occurs at $t=7\pi$.}
  \label{qw_p=1/4}
\end{figure}
In \cite{godsil2012state}, it is explained that perfect return will take place at $t=t_0$ if perfect state transfer takes place at time $t=\frac{t_0}{2}$. 
It should be noted however that perfect state transfer is not observed in this case even though perfect return occurs. 
It is shown in \cite{christandl2004perfect,vinet2012construct} that perfect state transfer takes place in the model associated to the ordinary Krawtchouk polynomials when the parameter $p$ is set to $p=\frac{1}{2}$. 
We shall similarly consider the case with $p=\frac{1}{2}$ for the model under study which is based on $X_2$-Krawtchouk polynomials.
When $N=5$ and $p=\frac{1}{2}$,  the matrix $M$ becomes
\begin{equation}
M=\begin{pmatrix}
0 & \frac{\sqrt{70}}{28} & 0 & \frac{1}{14} & 0 & 0 & 0\\
\frac{\sqrt{70}}{28} & 0 & \frac{3\sqrt{21}}{28} & 0 & \frac{\sqrt{7}}{28} & 0 & 0\\
0 & \frac{3\sqrt{21}}{28} & 0 & \frac{3\sqrt{30}}{28} & 0 & \frac{\sqrt{5}}{28} & 0 \\
\frac{1}{14} & 0 & \frac{3\sqrt{30}}{28} & 0 & \frac{5\sqrt{10}}{28} & 0 & 0\\
0 & \frac{\sqrt{7}}{28} & 0 & \frac{5\sqrt{10}}{28} & 0 & \frac{3\sqrt{15}}{28} & 0\\
0 & 0 & \frac{\sqrt{5}}{28} & 0 & \frac{3\sqrt{15}}{28} & 0 & \frac{\sqrt{14}}{28}\\
0 & 0 & 0 & 0 &0 & \frac{\sqrt{14}}{28}& 0
\end{pmatrix}
\end{equation}
and its eigenvalues are given by 
\begin{equation}
1,\quad \frac{3}{7},\quad \frac{1}{7},\quad 0,\quad -\frac{1}{7},\quad -\frac{3}{7},\quad -1.
\end{equation}
The plot of $\{ |c_{1j}(t)|\}_{j=1}^N$ is given in Fig. \ref{qw_p=1/2}.

\begin{figure}[htbp]
\begin{center}
\begin{tabular}{cl}
\includegraphics[width=8cm]{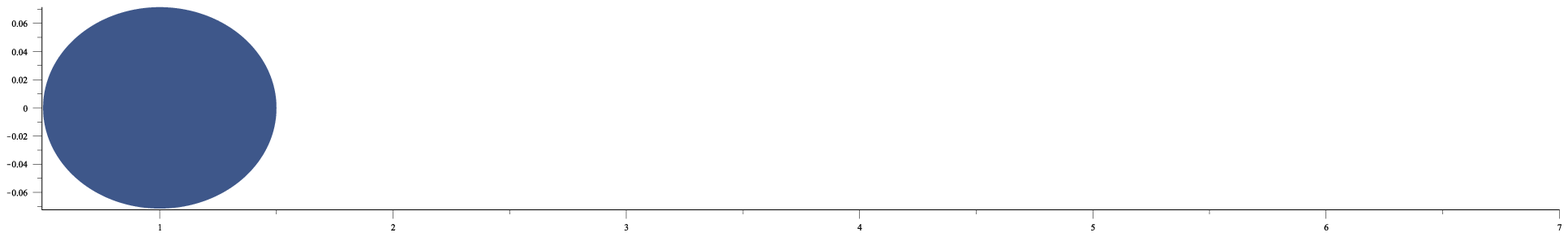}  & $(t=0)$\\
\includegraphics[width=8cm]{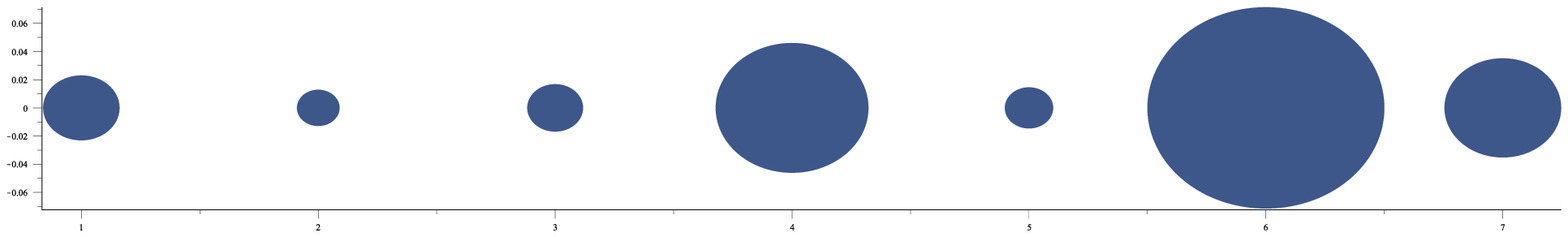} & $(t=\frac{7}{2}\pi)$\\
\includegraphics[width=8cm]{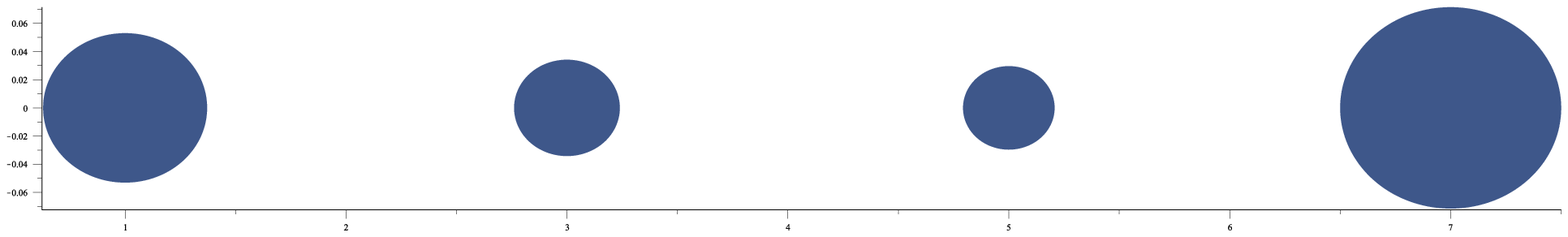} &  $(t=7\pi)$\\
\includegraphics[width=8cm]{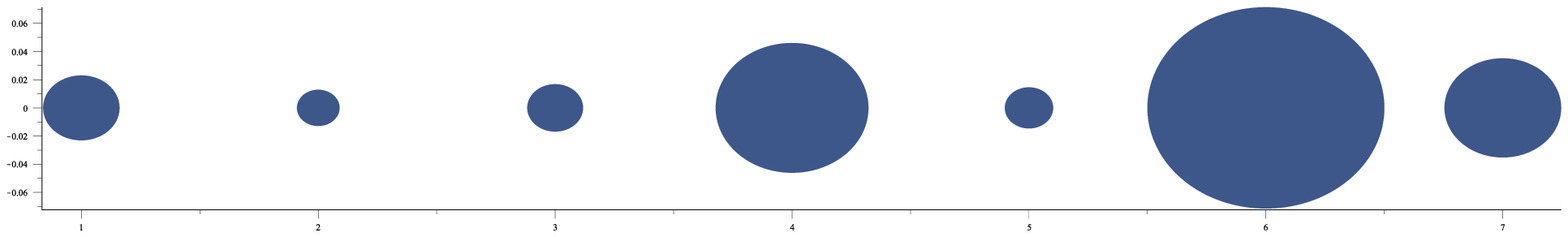} & $(t=\frac{21}{2}\pi)$\\
\includegraphics[width=8cm]{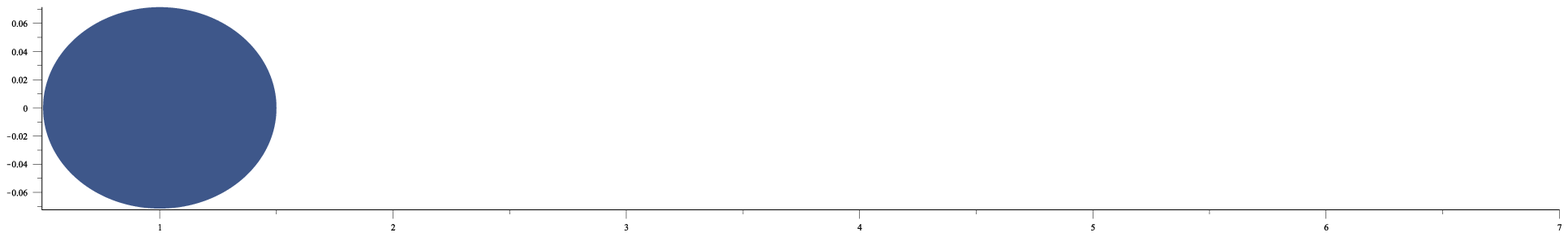} & $(t=14\pi)$
\end{tabular}
\end{center}
  \caption{The plot of transition amplitude $\{ |c_{1i}(t)|\}_{i=0}^N$ with $N=5$ and $p=\frac{1}{2}$. The areas of the circles are proportional to $|c_{1j}(t)|$ at the given lattice point $i$. Perfect return from $1$ to $1$ occurs at $t=14\pi$ and fractional revival from $1$ to odd sites is found at $t=7\pi$.}
  \label{qw_p=1/2}
\end{figure}
In this case, perfect state transfer is not observed either; interestingly however there is fractional revival, i.e. the site starting from $1$ is found only at odd sites at times $7\pi$ modulo $14\pi $.
This can be summarized in the following theorem.
\begin{theorem}\label{theorem:fr}
When $p=\frac{1}{2}$, perfect return takes place at some time $t_0$ such that
\begin{equation}
|c_{ij}(t_0)|=\delta_{ij}
\end{equation}
and 
\begin{equation}
\left| c_{ij}\left( \frac{t_0}{2}\right)\right|=0,\quad i+j:odd
\end{equation}
for $N \equiv 6 \mod 8$ and 
\begin{equation}
\left| c_{ij}\left( \frac{t_0}{2}\right)\right|=0,\quad i+j+N:even 
\end{equation}
for $N \not\equiv 6 \mod 8$.
\end{theorem}
In order to prove this theorem, we need to examine the eigenvalues $\{ \bar{\lambda}_x\}_{x=-1}^N$.

\begin{lemma}\label{lem:spec}
When $p=\frac{1}{2}$, the difference between two consecutive eigenvalues is given by
\begin{equation}
\mu_x = \bar{\lambda}_{x+1}-\bar{\lambda}_{x}=-\frac{N^2-4Nx+4x^2-3N+8x+6}{\binom{N+3}{3}},\quad x=-1,0,\ldots ,N-1.
\end{equation}
Furthermore, let us introduce the number $n_0$ defined by
\begin{equation}
n_0=\begin{cases}
\frac{\binom{N+3}{3}}{2} & N\equiv 0,3,4,7 \mod 8\\
\frac{\binom{N+3}{3}}{4} & N\equiv 1,2,6 \mod 8\\
\frac{\binom{N+3}{3}}{8} & N\equiv 5 \mod 8
\end{cases}
\end{equation}
Then the following holds:
\begin{enumerate}[(i)]
\item When $N\equiv 0,2,3,4,5,7\mod 8$, $n_0\mu_x$ is always integer and odd.
\item When $N\equiv 1,6\mod 8$, $n_0 \mu_x$ is always integer and its parity depends on $x$.
\end{enumerate}
\end{lemma}

\begin{proof}[Proof of Theorem. \ref{theorem:fr}]
We will show that
\begin{equation}
t_0=2n_0\pi.
\end{equation}
From the lemma \ref{lem:spec}, one can easily see that $2n_0\mu_x $ is always an even integer so that $e^{2t_0\mu_x i}=1$.
Therefore, much like in corollary \ref{cor:pr}, one has $|c_{ij}(t_0)|=\delta_{ij}$.
\begin{itemize}
\item When $N\equiv 0,2,3,4,5,7\mod 8$, from the lemma \eqref{lem:spec}, one observes that $n_0\mu_x$ is always an odd integer and that
\begin{equation}
e^{-i(\bar{\lambda}_{x+1}-\bar{\lambda}_x)\frac{t_0}{2}}=(-1)^{x-1},
\end{equation}
which results in
\begin{equation}
c_{ij}\left( \frac{t_0}{2} \right) = -e^{-i\bar{\lambda}_{-1}xt_0}\sum_{x=-1}^N T_i\left( x;\frac{1}{2}\right) T_j\left( x;\frac{1}{2}\right)(-1)^{x}.
\end{equation}
We introduce the function $F_{ij}(p)$ as follows
\begin{align}
\begin{split}
&F_{i+1,j+1}(p)=\sum_{x=-1}^N T_{i+1}(x;p)T_{j+1}(x;p) (-1)^x,\\
    &=\frac{1}{\sqrt{\hat{h}_{i}\hat{h}_{j}}}\sum_{x=-1}^N \binom{N+1}{x+1}\frac{p^{x+1}q^{N-x}(-1)^x}{f^p_l(x)f^p_l(x+1)}\hat{K}_{i}^{(2)}(x;p)\hat{K}_{j}^{(2)}(x;p).
\end{split}
\end{align}
From \eqref{kraw:sym} and \eqref{ekraw:sym}, one finds 
\begin{align}
\begin{split}
&F_{i+1,j+1}(q)=\frac{1}{\sqrt{\hat{h}_i\hat{h}_j}}\sum_{x=-1}^N \binom{N+1}{x+1}\frac{q^{x+1}p^{N-x}(-1)^x}{f^q_l(x)f^q_l(x+1)}\hat{K}_i^{(2)}(x;q)\hat{K}_j^{(2)}(x;q)\\
&=\left( -\frac{p}{q}\right)^{i+j} \frac{1}{\sqrt{\hat{h}_i\hat{h}_j}}\sum_{x=-1}^N \left\{ \binom{N+1}{x+1}\frac{q^{x+1}p^{N-x}(-1)^x}{f^q_l(N-x)f^q_l(N-x-1)}\right.\\
&\left.\qquad \qquad \qquad \hat{K}_i^{(2)}(N-x-1;q)\hat{K}_j^{(2)}(N-x-1;q)(-1)^x\right\}\\
&=(-1)^{N+i+j+1}\left( \frac{p}{q}\right)^{i+j}F_{i+1,j+1}(p).
\end{split}
\end{align}
Therefore, we have
\begin{equation}
F_{i+1,j+1}\left( \frac{1}{2}\right)= (-1)^{N+i+j+1}F_{i+1,j+1}\left( \frac{1}{2}\right)
\end{equation}
to find 
\begin{equation}
c_{ij}\left( \frac{t_0}{2} \right)=-e^{-i\bar{\lambda}_{-1}\frac{t_0}{2}}F_{ij}\left( \frac{1}{2}\right) = 0\quad (N+i+j:even).
\end{equation}
\item When $N\equiv 1,6\mod 8$, one notices that $e^{i\pi \mu_{x}} = (-1)^{x-1}$ does not hold. From the direct calculation, one can show that
\begin{equation}
\bar{\lambda}_x=-K^{-N-1}_3\left( x-N;\frac{1}{2}\right)=K^{-N-1}_3\left( -x-1;\frac{1}{2}\right)=-\bar{\lambda}_{N-x-1}
\end{equation}
and 
\begin{equation}
2\bar{\lambda}_x n_0=\frac{(2x+1-N)(N^2-4Nx+4x^2+N+4x+6)}{12}
\end{equation}
is always an even integer for $N\equiv 1,6\mod 8$.
We thus obtain 
\begin{equation}\label{spec:sym}
e^{-\bar{\lambda}_{N-x-1}\frac{t_0}{2}i}=e^{\bar{\lambda}_{x}\frac{t_0}{2}i}=e^{\bar{\lambda}_{x}t_0i}e^{-\bar{\lambda}_x \frac{t_0}{2}i}=e^{-\bar{\lambda}_{x}\frac{t_0}{2}i}.
\end{equation}
Here we consider 
\begin{equation}
\tilde{F}_{i+1,j+1}(p)=\sum_{x=-1}^N T_{i+1}(x;p)T_{j+1}(x;p)e^{-\bar{\lambda}_x\frac{t_0}{2}i}.
\end{equation}
In this case one finds from \eqref{spec:sym} that
\begin{equation}
\tilde{F}_{i+1,j+1}(q)=\left(-\frac{p}{q}\right)^{i+j}\tilde{F}_{i+1,j+1}(p)
\end{equation}
to derive
\begin{equation}
c_{i,j}\left( \frac{t_0}{2}\right)= \tilde{F}_{ij}\left( \frac{1}{2}\right)=0\quad (i+j:odd).
\end{equation}
\end{itemize}
Combining the above results completes the proof.
\end{proof}
This is illustrated in Fig. \ref{qw_p=1/2_N=6}.
\begin{figure}[htbp]
\begin{center}
\begin{tabular}{cl}
\includegraphics[width=8cm]{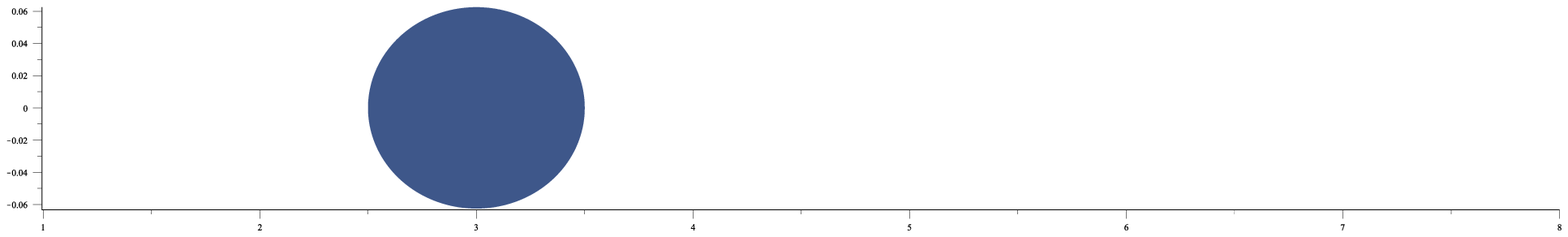}  & $(t=0)$\\
\includegraphics[width=8cm]{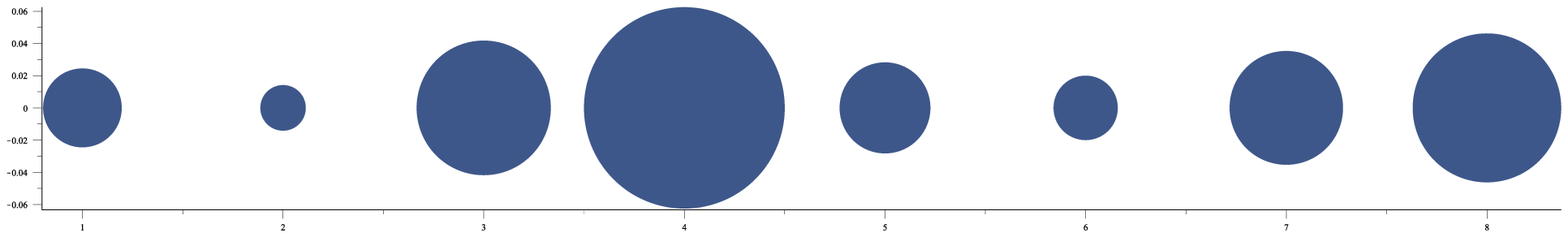} & $(t=\frac{21}{2}\pi)$\\
\includegraphics[width=8cm]{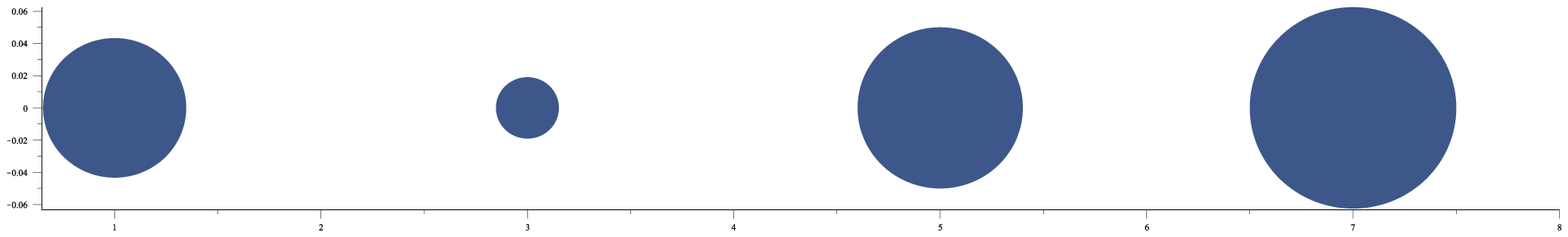} &  $(t=21\pi)$
\end{tabular}
\end{center}
  \caption{The plot of transition amplitude $\{ |c_{3i}(t)|\}_{i=0}^N$ with $N=6$ and $p=\frac{1}{2}$. The areas of the circles are proportional to $|c_{3j}(t)|$ at the given lattice point $j$. Fractional revival from $3$ to even sites occurs at $t=21\pi$.}
\end{figure}
\label{qw_p=1/2_N=6}
\section{Concluding remarks}

This paper has illustrated the type of generalized birth and death processes and continuous-time quantum walks that univariate exceptional orthogonal polynomials can underpin. To that end we focused on models connected to the exceptional $X_2$-Krawtchouk polynomials: a generalized birth and death process with states labelled by the integer set $\{0, 1,\dots, N, N+3\}$ and the corresponding quantum walk on the path weighted by the normalized recurrence coefficients of these polynomials. The salient and distinctive features of these problems have been discussed. The study presented here suggests the pursuit of various research questions.

It would obviously be of interest to similarly examine generalized BDP and quantum walks associated to other families of exceptional orthogonal polynomials. In dealing with ordinary polynomials, next to using the Krawtchouk set, the other simplest case is the one provided by the dual Hahn polynomials\cite{albanese2004mirror}. This brings to mind the idea of considering in this spirit applications of exceptional dual Hahn polynomials. A challenging problem is the construction of multivariate exceptional polynomials. Working this out would obviously open the way to new multidimensional extensions of BDP and quantum walks.

Another speculative question asks if the quantum walks on the paths weighted with the normalized recurrence coefficients of the exceptional orthogonal polynomials admit lifts to higher dimensional graphs. Studying the dynamics inferred by the polynomials from that perspective could prove revealing. See for instance \cite{christandl2005perfect, bernard2018graph} or \cite{miki2019quantum} in higher dimensions.

We shall recall that continuous-time quantum walks on weighted paths can be taken to define free-fermion chains with the time translation generator of the walk providing the couplings between the (second-quantized) fermions. A question that is raising much attention is the characterization of the entanglement within such chains \cite{eisler2018properties}. Advances have been made recently \cite{crampe2019free} by using the generalized Heun operators \cite{grunbaum2018algebraic} attached to the underlying (classical) orthogonal polynomials. Could this be extended to fermionic chains associated to exceptional polynomials such as the Krawtchouk ones?
We hope to report on some of these open questions in the near future.

\subsection*{Acknowledgement}
The authors would like to thank Akihiro Saito for dicussions.
The research of HM and ST is supported by JSPS KAKENHI (Grant Numbers 21H04073 and 19H01792 respectively) and that of LV by a discovery grant of the Natural Sciences and Engineering Research Council (NSERC) of Canada.

\bibliographystyle{unsrt}
\bibliography{ekraw}

\end{document}